\documentclass[10pt,conference]{IEEEtran}

\usepackage{amsfonts}
\usepackage{times}
\usepackage{latexsym}
\usepackage{amssymb}
\usepackage{amsmath}
\usepackage{cite}
\usepackage{verbatim}

\newcommand{\bydef}{\triangleq}

\def\bydef{:=}

\def\bb0{{\mathbb{0}}}

\def\bydef{:=}

\def\bb{{\mathbf{b}}}

\def\bh{{\mathbf{h}}}

\def\b0{{\mathbf{0}}}

\def\bB{{\mathbf{B}}}

\def\bbE{{\mathbb{E}}}

\def\bbN{{\mathbb{N}}}

\def\bbR{{\mathbb{R}}}

\def\bbZ{{\mathbb{Z}}}

\def\bydef{:=}

\def\sf0{{\mathsf{0}}}

\usepackage{graphicx}
\usepackage{amssymb}
\usepackage{amsfonts}
\usepackage{amsmath}
\usepackage{latexsym}
\usepackage{epstopdf}

\begin{document}

\newtheorem{thm}{Theorem}
\newtheorem{lemma}{Lemma}
\newtheorem{rem}{Remark}
\newtheorem{exm}{Example}
\newtheorem{prop}{Proposition}
\newtheorem{defn}{Definition}
\newtheorem{cor}{Corollary}
\def\proof{\noindent\hspace{0em}{\itshape Proof: }}
\def\endproof{\hspace*{\fill}~\QED\par\endtrivlist\unskip}
\def\bh{{\mathbf{h}}}
\def\SIR{{\mathsf{SIR}}}
\def\SINR{{\mathsf{SINR}}}

\title{Bounds on Minimum Number of Anchors for Iterative Localization and its Connections to Bootstrap Percolation}

\author{\authorblockN{Rahul Vaze}
\authorblockA{
Tata Institute of Fundamental Research\\
School of Technology and Computer Science\\
Homi Bhabha Road, Mumbai 400005\\
email: vaze@tcs.tifr.res.in}
\and
\authorblockN{Piyush Gupta}
\authorblockA{Bell Laboratories,  \\ 
Alcatel-Lucent, Murray Hill, \\
NJ 07974 USA \\ e-mail: pgupta@research.bell-labs.com}}

\date{}
\maketitle
\noindent
\begin{abstract} 
Iterated localization is considered where each node of a network needs to get localized (find its location on 2-D plane), when initially only a subset of nodes have their location information. The iterated localization process proceeds as follows. Starting with a subset of nodes that have their location information, possibly using global positioning system (GPS) devices, any other node gets localized if it has three or more localized nodes in its radio range. The newly localized nodes are included in the subset of nodes that have their location information for the next iteration. This process is allowed to continue, until no  new node can be localized. The problem is to find the minimum size of the initially localized subset to start with so that the whole network is localized with high probability. There are intimate connections between iterated localization and bootstrap percolation, that is well studied in statistical physics. Using results known in bootstrap percolation, we find a sufficient condition on the size of the initially localized subset that guarantees the localization of all nodes in the network with high probability.
\end{abstract}
\section{Introduction}
Several applications in wireless networks require that nodes know their location, e.g. vehicular or military networks. Localization is a technique to determine the physical coordinates of all nodes in the ad hoc network. These coordinates can be absolute, referenced through GPS, or relative, e.g. 
referenced through pair wise node distances. Absolute localization is typically achieved by having few anchor nodes who have their 
absolute location through some external means, e.g. GPS device. 
In the presence of anchor nodes, other nodes can acquire their location in the 2-D plane (or localize themselves) 
if there are three or more non-collinear anchor nodes in their radio range. The radio range is defined to be the maximum distance for which the received signal strength is above a required threshold. Most localization techniques are non-iterative in nature, where locations are obtained in one shot computation using the anchor nodes' locations. 
The non-iterative localization methods known in literature use received signal strength, hop count, time difference of arrival, semi-definite programming, multidimensional scaling etc. for finding node locations \cite{MooreLocalization2004, Biswas2004, Mao2006, BookStojmenovic}. 
The non-iterative methods do not exploit the fact that if any node that  has acquired its location information with sufficient accuracy it can aid in localizing few other nodes lying in its neighborhood that were not localized earlier. There has been extensive experimental work on iterative localization \cite{SavareseLoc2001}, anchor free distributed localization \cite{HariLoc2003}, mobile assisted localization \cite{HariLoc2005}, however, not many theoretical results are available. Necessary conditions have been derived in \cite{ErenLoc2004, AspnesLoc2006, Fang2009} for iteratively localizing an arbitrary network, together with a polynomial time algorithm.

In this paper we explore the iterative localization method for a network where $n$ nodes are randomly placed in a bounded area on a 2-D plane. To start with, a subset of nodes (called {\it anchors}) of size $m$ are given their location information. We assume that the $n$ nodes are uniformly distributed in the bounded area, and the anchors are chosen uniformly randomly from the $n$ nodes.
In each iteration, any node gets it location information (or gets {\it localized}) if it has three or more non co-linear localized nodes in its radio range $r(n)$. The newly localized nodes are included in the set of anchors for the next iteration. This process is allowed to continue, until no new node can be localized. The problem is to find the minimum number of the initial anchors  to start with so that the whole network is localized with high probability.

Iterated localization is similar to bootstrap percolation, which has been extensively  studied in statistical physics \cite{Lebowitz1989,Adler1991}. In bootstrap percolation, a 2-D $n \times n$ grid is considered, and each grid point is initially active with probability $p$ and not-active with probability $1-p$, independently of all other grid points. In each iteration, any grid point that is not active becomes active if it has two or more active neighbors out of its four nearest neighbors, and once a node becomes active it remains active forever. The  problem studied is to find the threshold on $p$ such that all grid points become active eventually with high probability.  For the 2-D case, sharp bounds on $p$ have been derived as a function of $n$ such that for $p \le p_c$, all grid points do not become active with high probability, while for $p > p_c$, all grid points are active eventually with high probability. The critical probability $p_c = \frac{\pi^2}{18 \ln n}$ \cite{HolroydBP2002}. Various generalizations of bootstrap percolation have also been studied, for example, on Erd{\H o}s-R{\'e}nyi graph \cite{JansonBP2010}, random geometric graph with a fixed radio range \cite{SanieeBP2012}, generalized bootstrap percolation \cite{BollobasBP2012} etc.

To use results from bootstrap percolation for solving the iterated localization problem, we connect the two problems as follows.
We map the iterated percolation problem to a suitable virtual grid, where each grid point becomes active if it has three or more active nodes out of the eight nearest neighbors. We then show that a sufficient condition for all the virtual grid points to become active eventually is identical to a sufficient condition for all the grid points for bootstrap percolation on 2-D to be active eventually.
We get the following sufficient condition on minimum $m$ (number of initial anchors) required for complete localization.

For radio range $r(n)$, and $\tau(n) =\Theta(r(n))$, such that $\lim_{n\rightarrow \infty }\left(1-e^{-n\pi \tau^2}\right)^{\frac{2}{r^2}} = 1$,   if \[ \left(1- \frac{\exp^{-m \pi r^2 } - \exp^{-n\pi r^2}}{1-\exp^{-n\pi r^2}}\right)> \frac{c}{\ln\left(\frac{\sqrt{2}}{r}\right)},\] where $c$ is a constant, then all nodes of the network are localized eventually with high probability. In particular, for $r(n) = \Theta \left(\sqrt{ \frac{{\ln n}}{n}}\right)$, $m = {\cal O}\left(n^{2.5/3}\right)$ (numerical result). 

\section {Notation:}
Let $S_1$ be a set and $S_2$ be a subset of $S_1$. Then $S_1 \backslash S_2$ denotes the set of elements of 
$S_1$ that do not belong to $S_2$. Cardinality of set $S$ is denoted by $|S|$. A disc of radius $r$ with center $x$ is denoted by $\bB(x,r) = \{y \in \bbR^2: |x-y|^2\le r\}$. 
Let $f(n)$ and $g(n)$ be two function defined on some subset of real numbers. 
Then we write $f(n) = \Omega(g(n))$ if
$\exists \ k > 0, \ n_0, \ \forall \ n>n_0$, $|g(n)| k \le |f(n)|$,
$f(n) = {\cal O}(g(n))$ if $\exists \ k > 0, \ n_0, \ \forall \ n>n_0$, $|f(n)| \le |g(n)| k$, and
$f(n) = \Theta(g(n))$ if $\exists \ k_1, \ k_2 > 0, \ n_0, \ \forall \ n>n_0$,
$|g(n)| k_1 \le |f(n)| \le |g(n)| k_2$. We use the symbol
$\bydef$  to define a variable.
 
\section{System Model}
Let $V$ be a set of $n$ nodes located in a unit ($1\times 1$) square in $\bbR^2$.
We assume that locations of nodes in $V$ are distributed as a  Poisson point process with density $n$ (approximating uniformly random node locations). Even though we consider a square, inherently we are assuming a toroidal or spherical surface, and we ignore the edge effects.
Let $A \subseteq V, |A|=m$ be the set of anchor nodes that have a GPS device, through which 
they exactly know their location on the 2-D plane. We assume that the set $A$ of $m$ anchor nodes is chosen uniformly randomly from the set $V$.
Each node of $V$ has a radio range of $r(n) <1$, i.e. each node can receive (transmit) transmissions  originating (send to) in a circle of radius $r(n)$ around it. For simplicity, we denote $r(n)$ as $r$. We assume that if a node has its own location information it can convey that information to all nodes within its radio range. 
Let at time $t$, $A_t \subseteq V$ be the set of the nodes that have their location information, $A_0 = A$.
We assume the following sequential location update rule. Using triangulation, at time $t+1$ any node $x \in V\backslash A_t$ can locate itself if there are $3$ or more nodes of $A_t$ in its radio range, i.e. $A_{t+1} = A_t \cup U_t$, where $U_t \bydef \{x \in V\backslash A_t : |\bB(x,r) \cap A_t| \ge 3 \}$. 
The problem is to find the minimum $m$ for which all $n$ nodes get located eventually with high probability, $$m^{\star}  =  \min_{\lim_{n\rightarrow \infty}P(\lim_{t\rightarrow \infty}A_{t} = n) =1} m.$$ 

There is an intimate connection between iterated localization considered in this paper and bootstrap percolation \cite{Lebowitz1989}. We briefly review the bootstrap percolation model and some results  that are useful for our analysis.

\subsection{Preliminaries on Bootstrap Percolation}
Consider a $n\times n$ regular grid $R$ on $\bbZ^2$ with side-length $1$.  Let at time step $t=0$ any grid point of $R$ be active (inactive) with probability $ p(n)$ $(1-p(n))$ independently of all other grid points. In subsequent time steps an inactive grid point becomes active if two or more of its four neighbors are active. Once a grid point becomes active it stays active forever. This model of sequential activation is called {\it bootstrap percolation}. The problem that has been extensively studied in bootstrap percolation is on finding the critical value of $p_c(n)$ such that all $n^2$ grid points are active eventually, i.e. 
for $ p(n)> p_c(n)$, $\lim_{n\rightarrow \infty}P(\text{all} \ n^2 \ \text{grid points are active eventually} )=1$, while for $ p(n) < p_c(n)$, $\lim_{n\rightarrow \infty}P(\text{all} \ n^2 \ \text{grid points are active eventually} ) <1$. In the following Lemma we summarize the result.

\begin{thm} \cite{Lebowitz1989} For bootstrap percolation on a  $n\times n$ regular grid $R$ on $\bbZ^2$, $p_c(n) = \Theta \left(\frac{1}{\ln n}\right)$.
\end{thm}

Next, we present a sufficient condition derived in \cite{Lebowitz1989}  used to show the sufficiency of $p_c(n) = {\cal O}\left(\frac{1}{\ln n}\right)$ for bootstrap percolation.

\begin{lemma}\label{lem:suff}If initially at time $t=0$, each of the faces of the square with side $(2k+1), \ k=0,1,\dots, n/2$, have at least one active node,
then all the $n^2$ vertices of the grid are active eventually.
\end{lemma}

\begin{proof}
Consider the smallest square $S_1$ of side $1$ at the center of the grid. As long as there are two active grid points in $S_1$, all four nodes of $S_1$ are active. Now, assume that each grid point of square $S_{\ell}, \ell \ge 1$ (square with side $\ell$ around the center) is active at time $t$. Then if there is at least one active grid points on each of the faces of $S_{\ell+1}$, then clearly, all grid points of $S_{\ell+1}$ become active in next time step $t+1$.
\end{proof}
\begin{thm}\label{thm:bp}\cite{Lebowitz1989} For bootstrap percolation on a  $n\times n$ regular grid $R$ on $\bbZ^2$, $p_c(n) ={\cal O}\left(\frac{1}{\ln n}\right)$.
\end{thm}
\begin{proof} If $p(n)= {\cal O}\left(\frac{1}{\ln n}\right)$, then each of the faces of the square with side $(2k+1), \ k=0,1,\dots, n/2$, have at least one active node with high probability. Then using Lemma \ref{lem:suff} we get the required result.
\end{proof}
We will use Theorem \ref{thm:bp} for finding a sufficient condition for localizing all nodes eventually with iterated localization.


\subsection{Upper Bound on $m^{\star}$} We construct a virtual grid $G$ from  the $1\times 1$ square with side length $\frac{r}{\sqrt{2}}$ as shown in Fig. \ref{fig:lattice}. For each grid point $g_{ij} \bydef (i,j)$, 
consider a ball of radius $\tau = \tau(n)$ around $g_{ij}$, $\bB(g_{ij},\tau)$, where $\tau < \frac{r}{2\sqrt{2}}$. 
We define $E_{g_{ij}}=1$ if $|\bB(g_{ij},\tau) \cap V|\ge 1$, i.e. if there is a node of $V$ within a ball of radius $\tau$ from $g_{i,j}$, and $E_{g_{ij}}=0$ otherwise. 
Since $\tau < \frac{r}{2\sqrt{2}}$, events $E_{g_{ij}}=1$ are independent  $\forall \ i,j$. The probability that for each of the grid points $E_{g_{ij}}=1, \ \forall \ i,j$, is 
$P(\cap_{ \forall  i,j} E_{g_{ij}}=1) = \left(1-e^{-n\pi \tau^2}\right)^{\frac{2}{r^2}}$. 
Let $r$ and $\tau$ be such that $n \rightarrow \infty$, $P(\cap_{ \forall  i,j} E_{g_{ij}}=1) \rightarrow 1$, i.e. $E_{g_{ij}}=1, \ \forall \ i,j$ with high probability.

\begin{figure}
\centering
\includegraphics[width=3.4in]{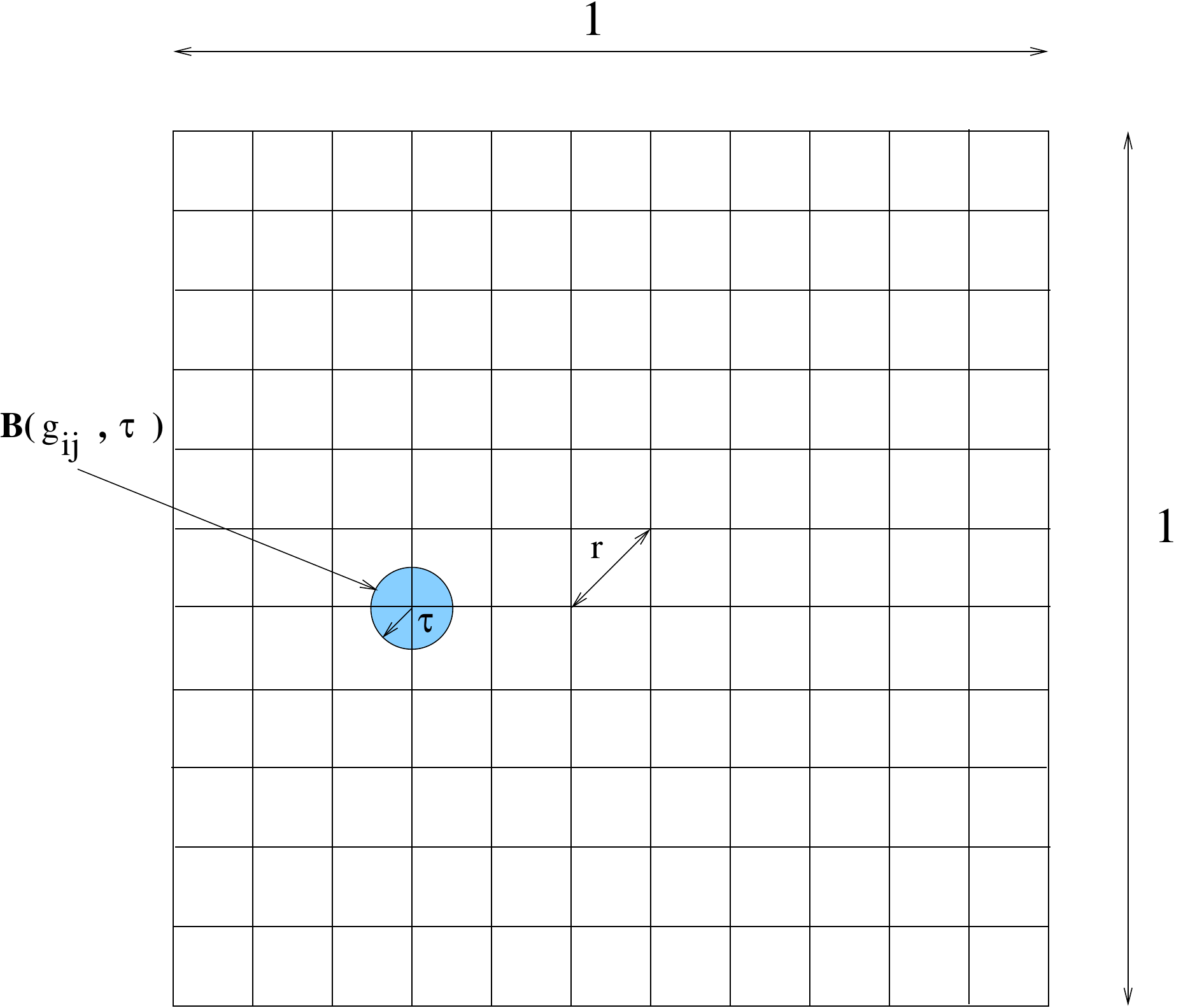}
\caption{Virtual Grid}
\label{fig:lattice}
\end{figure}




%

%
Next, we will map iterated localization to bootstrap percolation on the virtual grid $G$ as follows to use the results known in bootstrap percolation, Theorem \ref{thm:bp}. 
Recall that we have chosen $r$ and $\tau$ such that for large enough $n$, $E_{g_{ij}}=1, \ \forall \ i,j$ with high probability, i.e. with high probability there is at least one node in $\bB(g_{ij},\tau), \ \forall \ g_{i,j}$.  
For the rest of the discussion  we condition on the event that $E_{g_{ij}}=1, \ \forall \  i,j$.
We define a virtual vertex $g_{ij}$ of the virtual grid $G$ to be {\it red} if there is at least one GPS node (out of total $m$) in $\bB(g_{ij},\tau)$, otherwise its {\it blue}. Since $\tau < \frac{r}{2\sqrt{2}}$, this coloring is independent for all vertices. 
Let $q \bydef P(g_{ij} \ \text{is initially red}), \ \forall \ i,j$. 
We define the virtual grid bootstrap percolation where  any blue virtual vertex $g_{ij}$ becomes red if it has  three or more virtual red neighbors out of its $8$ closest neighbors on the virtual grid. The updation process is allowed to continue, where once a virtual vertex becomes red it remains red forever. 
This is a different update condition  compared to the one defined earlier for bootstrap percolation, where only two of the four nearest neighbors were required for activation.


Consider a virtual vertex $g_{ij}$ that is blue initially. At any time step $t$, let $g_{ij}$ become red in the virtual grid bootstrap percolation, then $z \in V\cap \bB(g_{ij},\tau)$ can find three or more localized nodes in $\bB(z,r+2\tau)$. That is, if the radio range for iterated localization is enhanced to $r'=r+2\tau$,  $g_{ij}$ becoming red implies that all nodes belonging to $V\cap \bB(g_{ij},\tau)$ become localized. Moreover, if all virtual grid vertices are red eventually, then clearly for all nodes $v \in V$, $\bB(v,r+2\tau)$ contains at least three localized nodes if radio range is $r'$. Thus, if the radio range is $r+2\tau$, then the whole network $V$ of $n$ nodes is localized if the whole virtual grid is colored red eventually. 




%


Next, we  show that the sufficient condition for having each vertex to be red eventually in virtual grid bootstrap percolation (requiring three red neighbors out of eight) is same as that of bootstrap percolation (requiring two red neighbors out of four). 

\begin{lemma}\label{lem: sufcondlocalization} For bootstrap percolation on the virtual grid, where at least three out of eight immediate neighbors are required to be red for a vertex to be colored red, if initially each of the faces of the square with side $(2k+1)\frac{r}{\sqrt{2}}, \ k=0,1,\dots, \frac{\sqrt{2}}{r}$, have at least one red vertex, then all $\frac{2}{r^2}$ vertices of the grid are red eventually.
\end{lemma}
\begin{proof} Consider both $S_1$ and $S_3$, where $S_{\ell}$ is the square of side 
$\left(\frac{\ell  r}{\sqrt{2}}\right),  \ell = 2k+1, k\in \bbN$  as shown in Fig. \ref{fig:suff3on8}, where there is at most one red vertex in any of the faces of $S_1$ or $S_3$. For any configuration that has at least one red vertex on each face of $S_1$ and $S_3$, one can see that each blue vertex of $S_1\cup S_3$ eventually finds three or more red neighboring vertices and hence becomes red. Hence, eventually, vertices of $S_1$ and $S_3$ will become red if there is at least one red vertex on each face of $S_1$ and $S_3$.
 
Now we will prove the Lemma by using induction. Assume that all vertices of $S_{\ell}$ are red, and each face of $S_{\ell+1}$ has at least one red vertex. We know that all vertices of $S_{3}$ are red, hence induction can start. We will show that under this condition, all vertices of $S_{\ell+1}$ eventually become red.
Consider any blue vertex of $S_{\ell+1}$ that has a red vertex of $S_{\ell+1}$ as its neighbor. In particular consider blue vertex $v_1$ in Fig. \ref{fig:suff3on8expanded}. Clearly, $v_1$ has three red neighbors and hence it will become red in the next step. Continuing this process, all non-corner vertices of $S_{\ell+1}$ that lie on the same face as $v_1$  will find 
at least three red neighbors and will eventually become red themselves. Since each face of $S_{\ell+1}$ has at least one red vertex, all non-corner vertices of $S_{\ell+1}$ will eventually become red. Once all the non-corner vertices of $S_{\ell+1}$ are red, then even the corner vertices have three red neighbors, and consequently all vertices of $S_{\ell+1}$ will eventually become red. 
\end{proof}
\begin{figure}
\centering
\includegraphics[width=1.4in]{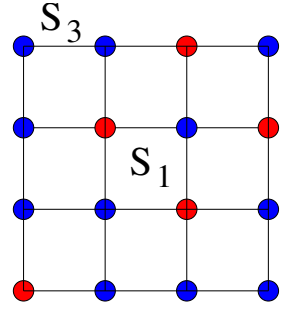}
\caption{Sufficient condition for making all grid points red. }
\label{fig:suff3on8}
\end{figure}
\begin{figure}
\centering
\includegraphics[width=1.4in]{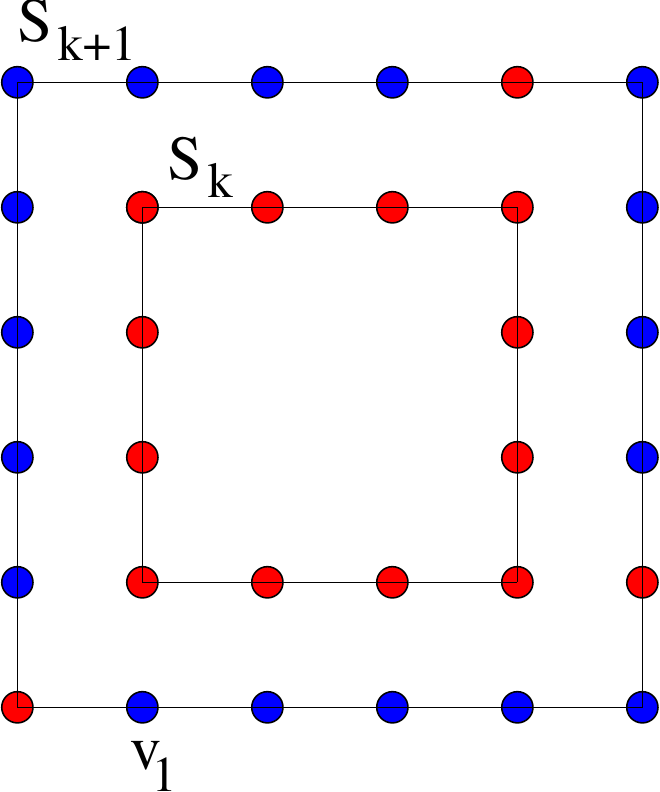}
\caption{Sufficient condition for making all grid points red. }
\label{fig:suff3on8expanded}
\end{figure}

Thus using Theorem \ref{thm:bp}, we get the following Theorem for the iterated localization.
\begin{thm}\label{thm:il}  
If the radio range is $r' = r+2\tau$, where $r$ and $\tau$ is such that 
$\lim_{n\rightarrow \infty}P(\cap_{ \forall  i,j} E_{g_{ij}}=1) = \lim_{n\rightarrow \infty}\left(1-e^{-n\pi \tau^2}\right)^{\frac{2}{r^2}}=1$, and 
$q = {\cal O}\left( \frac{1}{\ln(\frac{\sqrt{2}}{r})}\right)$, then all nodes of $V$ are localized eventually with high probability. 
\end{thm}
\begin{proof} Since the sufficient conditions in Lemma \ref{lem:suff} and Lemma \ref{lem: sufcondlocalization} are identical, the result follows from Theorem \ref{thm:bp}.
\end{proof}

Therefore, initially, if the probability of any vertex being red (having a GPS node) $q$ is more than 
 $\frac{c}{\ln(\frac{\sqrt{2}}{r})}$, where $c$ is a constant, then with high probability all the vertices of the grid are red with iterated localization. Using Theorem \ref{thm:il}, in the next corollary we find the sufficient number of anchor nodes $m$ to start with so that all nodes are localized eventually with high probability.


\begin{cor}\label{cor:scaled} For any $r(n)$ and $\tau(n) =\Theta(r(n))$, such that $\lim_{n\rightarrow \infty }P(\cap_{ \forall  i,j} E_{g_{ij}}=1) = \lim_{n\rightarrow \infty }\left(1-e^{-n\pi \tau^2}\right)^{\frac{2}{r^2}} = 1$,  then if $$ \left(1- \frac{\exp^{-m \pi r^2 } - \exp^{-n\pi r^2}}{1-\exp^{-n\pi r^2}}\right)> \frac{c}{\ln(\frac{\sqrt{2}}{r})}$$
 then all nodes of $V$ are localized eventually with high probability.
\end{cor}
\begin{proof} For any vertex $g_{ij}$, given that $E_{g_{ij}}=1$, 
we compute the probability $1-q$ that the vertex is {\it blue} (no GPS node out of total $ \text {  \\ 
$m$ lies in  $\bB(g_{ij},\tau)$}$) as follows. By definition, $1-q$
\begin{eqnarray*}
&= & P(| A_0 \cap \bB(g_{ij},\tau)| =0   | E_{g_{ij}}=1), \\
&=&  \frac{P(| A_0 \cap \bB(g_{ij},\tau)| =0, E_{g_{ij}}=1)}{P(E_{g_{ij}}=1)},\\
&=&  \frac{P(| A_0 \cap \bB(g_{ij},\tau)| =0, |V\cap \bB(g_{ij},\tau)|\ge 1)}{P(|V\cap \bB(g_{ij},\tau)|\ge 1)},\\
&=&  \frac{\bbE_{k\ge 1}\left\{P(| A_0 \cap \bB(g_{ij},\tau)| =0, |V\cap \bB(g_{ij},\tau)|= k)\right\}}{P(|V\cap \bB(g_{ij},\tau)|\ge 1)},\\
&=& \frac{\sum_{k=1}^{\infty} \left(1-\frac{m}{n}\right)^k \frac{(n\pi r^2)^k}{k!} \exp^{-n\pi r^2}}{1-\exp^{-n\pi r^2}},\\
&=& \frac{\exp^{-n\pi r^2 \frac{m}{n}} - \exp^{-n\pi r^2}}{1-\exp^{-n\pi r^2}},\\
&=& \frac{\exp^{-m \pi r^2 } - \exp^{-n\pi r^2}}{1-\exp^{-n\pi r^2}},
\end{eqnarray*}
and hence 
$q = \left(1- \frac{\exp^{-m \pi r^2 } - \exp^{-n\pi r^2}}{1-\exp^{-n\pi r^2}}\right)$. Thus, from Thoerem \ref{thm:il} for  $\tau(n) =\Theta(r(n))$ if  $ \left(1- \frac{\exp^{-m \pi r^2 } - \exp^{-n\pi r^2}}{1-\exp^{-n\pi r^2}}\right)> \frac{c}{\ln(\frac{\sqrt{2}}{r})}$, then all $n$ nodes of $V$ 
can be localized with radio range $r+2\tau = \Theta(r)$.


\end{proof}

\begin{rem} Recently, bootstrap percolation on random geometric graphs has been studied in \cite{SanieeBP2012}, where $n$ nodes are distributed uniformly in a $1\times 1$ square, however, for a fixed radio range $r(n) =\Theta \left(\sqrt{\frac{\ln n}{n}}\right)$. The radio range $r(n) = \Theta \left(\sqrt{\frac{\ln n}{n}}\right)$ is chosen to be equal to the connectivity radio range that ensures that there is a path between any two nodes in the network where a link is present between two nodes if their centers are less than $r(n)$ distance away \cite{Gupta1998}. To start with, each node is made active with probability $p$ and inactive otherwise, independently of all other nodes. Under this model, a node becomes active if there are more than $\theta = \gamma a \ln n$ for $a>1, \gamma > 0$ active nodes inside its radio range $r(n)$. Bounds on the critical probability have been derived in \cite{SanieeBP2012}. Note that bootstrap percolation on random geometric graph is identical to iterative localization with $\theta =3$ in our model. The result of \cite{SanieeBP2012} is obtained using a different approach compared to this paper, and its not clear how they can be applied for fixed $\theta =3$. \end{rem}

\begin{exm}\label{ex:one} Let $r(n) = \Theta \left(\sqrt{ \frac{{c \ln(n)}}{n}}\right)$ be the connectivity radius. 
Then one can check that  $\lim_{n\rightarrow \infty }\left(1-e^{-n\pi \tau^2}\right)^{\frac{2}{r^2}} = 1$ for $\tau < \frac{r}{2\sqrt{2}}$ and large enough $c$, hence  satisfying the condition in Corollary \ref{cor:scaled}. Thus, using Corollary \ref{cor:scaled} we get the minimum $m$ that is sufficient to get the whole network localized eventually with high probability. 
We plot the scaling of $m$ with respect to $n$ for different values of $c$ in Fig. \ref{fig:conrad} which suggets that $m$ scales as $n^{2.5/3}$ for most values of $c$. 
\end{exm}


\section{Conclusions}
In this paper we demonstrated connections between iterated localization and bootstrap percolation. Bootstrap percolation has been a well studied topic, and strong results are known for critical thresholds for many models. We mapped the iterated localization problem to an instance of bootstrap percolation problem, and using known results on bootstrap percolation obtained sufficient conditions on the minimum number of anchor nodes to start with that guarantee that the whole network is localized with high probability. 
Iterated localization is relatively less explored area of research and major emphasis has been on non-iterated localization methods. The major advantage of iterated localization method is the significant reduction in the number of GPS nodes to be deployed so that each node of the network can get localized. 
There is, however, an accuracy cost associated with iterated methods because of error propagation. For example, if the radio ranges of nodes are not accurate, each node which localizes itself can make errors, and since all localized nodes are used as  anchor nodes in subsequent iterations, error propagation can lead to significant localization inaccuracies.

\begin{figure}
\centering
\includegraphics[width=3.5in]{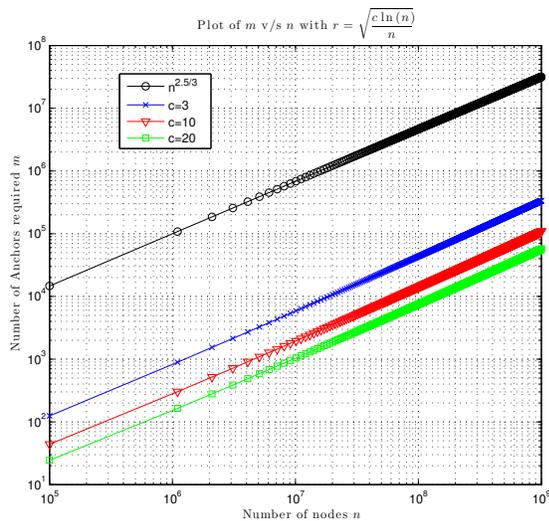}
\caption{Sufficient number of anchor nodes for localization as a function of total number of nodes $n$. }
\label{fig:conrad}
\end{figure}

\bibliographystyle{../../IEEEtran}
\bibliography{../../IEEEabrv,../../Research}

\begin{thebibliography}{10}
\providecommand{\url}[1]{#1}
\csname url@samestyle\endcsname
\providecommand{\newblock}{\relax}
\providecommand{\bibinfo}[2]{#2}
\providecommand{\BIBentrySTDinterwordspacing}{\spaceskip=0pt\relax}
\providecommand{\BIBentryALTinterwordstretchfactor}{4}
\providecommand{\BIBentryALTinterwordspacing}{\spaceskip=\fontdimen2\font plus
\BIBentryALTinterwordstretchfactor\fontdimen3\font minus
  \fontdimen4\font\relax}
\providecommand{\BIBforeignlanguage}[2]{{%
\expandafter\ifx\csname l@#1\endcsname\relax
\typeout{** WARNING: IEEEtran.bst: No hyphenation pattern has been}%
\typeout{** loaded for the language `#1'. Using the pattern for}%
\typeout{** the default language instead.}%
\else
\language=\csname l@#1\endcsname
\fi
#2}}
\providecommand{\BIBdecl}{\relax}
\BIBdecl

\bibitem{MooreLocalization2004}
\BIBentryALTinterwordspacing
D.~Moore, J.~Leonard, D.~Rus, and S.~Teller, ``{Robust distributed network
  localization with noisy range measurements},'' pp. 50--61, 2004. [Online].
  Available: \url{http://dx.doi.org/10.1145/1031495.1031502}
\BIBentrySTDinterwordspacing

\bibitem{Biswas2004}
\BIBentryALTinterwordspacing
P.~Biswas and Y.~Ye, ``{Semidefinite Programming for Ad Hoc Wireless Sensor
  Network Localization},'' pp. 46--54, 2004. [Online]. Available:
  \url{http://dx.doi.org/10.1109/IPSN.2004.1307322}
\BIBentrySTDinterwordspacing

\bibitem{Mao2006}
\BIBentryALTinterwordspacing
G.~Mao, B.~Fidan, and B.~Anderson, ``{Wireless sensor network localization
  techniques},'' \emph{Computer Networks}, vol.~51, no.~10, pp. 2529--2553,
  Jul. 2007. [Online]. Available:
  \url{http://dx.doi.org/10.1016/j.comnet.2006.11.018}
\BIBentrySTDinterwordspacing

\bibitem{BookStojmenovic}
I.~Stojmenovic, \emph{Handbook of Sensor Networks: Algorithms and
  Architectures}.\hskip 1em plus 0.5em minus 0.4em\relax Wiley, 2005.

\bibitem{SavareseLoc2001}
C.~Savarese, J.~Rabaey, and J.~Beutel, ``Location in distributed ad-hoc
  wireless sensor networks,'' in \emph{IEEE International Conference on
  Acoustics, Speech, and Signal Processing, 2001. Proceedings. (ICASSP '01).
  2001}, vol.~4, 2001, pp. 2037 --2040 vol.4.

\bibitem{HariLoc2003}
\BIBentryALTinterwordspacing
N.~B. Priyantha, H.~Balakrishnan, E.~Demaine, and S.~Teller, ``Anchor-free
  distributed localization in sensor networks,'' \emph{Science}, vol.~8, no.
  SenSys, pp. 340--341, 2003. [Online]. Available:
  \url{http://en.scientificcommons.org/43019921}
\BIBentrySTDinterwordspacing

\bibitem{HariLoc2005}
N.~Priyantha, H.~Balakrishnan, E.~Demaine, and S.~Teller, ``Mobile-assisted
  localization in wireless sensor networks,'' in \emph{Proceedings IEEE INFOCOM
  2005. 24th Annual Joint Conference of the IEEE Computer and Communications
  Societies.}, vol.~1, Mar. 2005, pp. 172 -- 183 vol. 1.

\bibitem{ErenLoc2004}
T.~Eren, O.~Goldenberg, W.~Whiteley, Y.~Yang, A.~Morse, B.~Anderson, and
  P.~Belhumeur, ``Rigidity, computation, and randomization in network
  localization,'' in \emph{Proceedings IEEE INFOCOM 2004. Twenty-third Annual
  Joint Conference of the IEEE Computer and Communications Societies}, vol.~4,
  Mar. 2004, pp. 2673 -- 2684 vol.4.

\bibitem{AspnesLoc2006}
J.~Aspnes, T.~Eren, D.~Goldenberg, A.~Morse, W.~Whiteley, Y.~Yang, B.~Anderson,
  and P.~Belhumeur, ``A theory of network localization,'' \emph{{IEEE} Trans.
  Mobile Comput.}, vol.~5, no.~12, pp. 1663 --1678, Dec. 2006.

\bibitem{Fang2009}
\BIBentryALTinterwordspacing
J.~Fang, M.~Cao, A.~S. Morse, and B.~D.~O. Anderson, ``Sequential localization
  of sensor networks,'' \emph{SIAM Journal on Control and Optimization},
  vol.~48, no.~1, pp. 321--350, 2009. [Online]. Available:
  \url{http://link.aip.org/link/?SJC/48/321/1}
\BIBentrySTDinterwordspacing

\bibitem{Lebowitz1989}
M.~Aizenman and J.~Lebowitz, ``Metastability effects in bootstrap percolation
  models,'' \emph{Journal of Physics A}, pp. 22:L297--L301, 1989.

\bibitem{Adler1991}
J.~Adler, ``Bootstrap percolation,'' \emph{Physica A}, pp. 171:453--470, 1991.

\bibitem{HolroydBP2002}
A.~E. {Holroyd}, ``{Sharp Metastability Threshold for Two-Dimensional Bootstrap
  Percolation},'' \emph{Probability Theory and Related Fields}, vol. 125(2),
  pp. 195--224.

\bibitem{JansonBP2010}
\BIBentryALTinterwordspacing
S.~{Janson}, T.~{{\L}uczak}, T.~{Turova}, and T.~{Vallier}, ``{Bootstrap
  percolation on the random graph G\_$\{$n,p$\}$},'' \emph{ArXiv e-prints},
  Dec. 2010. [Online]. Available:
  \url{http://adsabs.harvard.edu/abs/2010arXiv1012.3535J}
\BIBentrySTDinterwordspacing

\bibitem{SanieeBP2012}
M.~{Bradonji{\'c}} and I.~{Saniee}, ``{Bootstrap Percolation on Random
  Geometric Graphs},'' \emph{ArXiv e-prints}, Jan. 2012.

\bibitem{BollobasBP2012}
\BIBentryALTinterwordspacing
B.~{Bollob{\'a}s}, P.~{Smith}, and A.~{Uzzell}, ``{Generalized bootstrap
  percolation},'' \emph{ArXiv e-prints}, Apr. 2012. [Online]. Available:
  \url{http://adsabs.harvard.edu/abs/2012arXiv1204.3980B}
\BIBentrySTDinterwordspacing

\bibitem{Gupta1998}
P.~Gupta and P.~Kumar, ``Critical power for asymptotic connectivity in wireless
  networks,'' in \emph{Stochastic Analysis, Control, Optimization and
  Applications: A Volume in Honor of W.H. Fleming, W. M. McEneaney, G. Yin, and
  Q. Zhang (Eds.)}.\hskip 1em plus 0.5em minus 0.4em\relax Birkhauser, Boston,
  1998.

\end{thebibliography}

\end{document}